\documentclass[12pt]{article}
\usepackage{amsmath, amsthm, amssymb, amstext}
\usepackage{mathrsfs}
\usepackage{array, amsfonts, mathrsfs}
\usepackage{latexsym, amsmath}
\usepackage{graphicx,color}
\usepackage{epsf}
\usepackage{enumerate}

\newtheorem{Lemma}{Lemma}

\newtheorem{prop}{Proposition}

\date{}

\begin{document}

\title{Electric Impedance Tomography problem for surfaces with internal holes}

\author{A.V. Badanin\thanks {St.Petersburg State University, St.Petersburg, Russia, e-mail: a.badanin@spbu.ru.},\,
        M.I.Belishev\thanks {St.Petersburg Department of Steklov Mathematical
        Institute, St.Petersburg, Russia, e-mail: belishev@pdmi.ras.ru; supported by RFBR grant 20-01 627A}.\,
        D.V.Korikov\thanks {St.Petersburg Department of Steklov Mathematical
        Institute, St. Petersburg, Russia, e-mail: thecakeisalie@list.ru; supported by RFBR grant 20-01 627A}.}

\maketitle

\begin{abstract}
Let $(M,g)$ be a smooth compact Riemann surface with the
multicomponent boundary
$\Gamma=\Gamma_0\cup\Gamma_1\cup\dots\cup\Gamma_m=:\Gamma_0\cup\tilde\Gamma$.
Let $u=u^f$ obey $\Delta u=0$ in $M$,
$u|_{\Gamma_0}=f,\,\,u|_{\tilde\Gamma}=0$ (the grounded holes) and
$v=v^h$ obey $\Delta v=0$ in $M$,
$v|_{\Gamma_0}=h,\,\,\partial_\nu v|_{\tilde\Gamma}=0$ (the
isolated holes). Let $\Lambda_{g}^{\rm gr}: f\mapsto\partial_\nu
u^f|_{\Gamma_{0}}$ and $\Lambda_{g}^{\rm is}: h\mapsto\partial_\nu
v^h|_{\Gamma_{0}}$ be the corresponding DN-maps. The EIT problem
is to determine $M$ from $\Lambda_{g}^{\rm gr}$ or
$\Lambda_{g}^{\rm is}$.

To solve it, an algebraic version of the BC-method is applied. The
main instrument is the algebra of holomorphic functions on the
ma\-ni\-fold ${\mathbb M}$, which is obtained by gluing two
examples of $M$ along $\tilde{\Gamma}$. We show that this algebra
is determined by $\Lambda_{g}^{\rm gr}$ (or $\Lambda_{g}^{\rm
is}$) up to isometric isomorphism. Its Gelfand spectrum (the set
of characters) plays the role of the material for constructing a
relevant copy $(M',g',\Gamma_{0}')$ of $(M,g,\Gamma_{0})$. This
copy is conformally equivalent to the original, provides
$\Gamma_{0}'=\Gamma_{0},\,\,\Lambda_{g'}^{\rm gr}=\Lambda_{g}^{\rm
gr},\,\,\Lambda_{g'}^{\rm is}=\Lambda_{g}^{\rm is}$, and thus
solves the problem.
\end{abstract}

\noindent{\bf Key words:}\,\,\,determination of Riemann surface
from its DN-map, algebraic version of Boundary Control method.

\bigskip

\subsubsection*{About the paper}
\noindent$\bullet$\,\,\,The fact that the Dirichlet-to-Neumann map
determines the Riemannian surface with boundary up to conformal
equivalence, is now well known \cite{LUEns,BCald,HMich} and first
established in \cite{LUEns}. The approach
\cite{BCald,BKor_JIIPP,{BKor_char_arXiv}} is based on connections
between the EIT problem and Banach algebras of holomorphic
functions. The goal of our paper is to extend this approach to
surfaces with a multicomponent boundary $\Gamma$, provided that
only one connected component $\Gamma_{0}$ of $\Gamma$ is
accessible for the measurements and observations.
\smallskip

\noindent$\bullet$\,\,\,Determination of unknown boundaries is a
well-known problem, and the literature devoted to it is hardly
observable (see, e.g., \cite{Aless1,Aless2}). The specific feature
of the statement that we deal with, is the following. In the
mentioned papers, the surface/domain $M$ and its external boundary
$\Gamma\subset\partial M$ are assumed to be {\it given}, as well
as the parameters of the media (conductivity, density, etc), which
fills $M$. Roughly speaking, we know the surface but do not know
the holes into it. The aim is to determine $\partial
M\setminus\Gamma$ (the holes) from the measurements on $\Gamma$.
In contrast to such a traditional set up, we do not assume the
surface to be known and recover {\it the surface together with
holes} on it.

\subsubsection*{Statement of problem and results}

\noindent$\bullet$\,\,\,Let $(M,g)$ be a compact oriented Riemann
surface with the boundary
$\Gamma=\Gamma_0\cup\Gamma_1\cup\dots\cup\Gamma_m=:\Gamma_0\cup\tilde\Gamma$,
each $\Gamma_j$ being diffeomorphic to a circle in
$\mathbb{R}^{2}$; $g$ the smooth \footnote{everywhere in the
paper, {\it smooth} means $C^{\infty}$-smooth} metric tensor;
$\Delta_{g}$ the Beltrami-Laplace operator.

Consider the problem
\begin{equation}
\label{problem gr} \Delta_{g} u=0\,\,\, \text{ in }\,\, {\rm
int}M,\qquad u=f\,\,\,\text{ on } \Gamma_{0}, \qquad u=0
\,\,\,\text{ on } \tilde{\Gamma}
\end{equation}
(${\rm int\,}M:=M\setminus\Gamma$) and denote by $u^{f}$ the
solution of (\ref{problem gr}) for a smooth $f$. With problem
(\ref{problem gr}) one associates the Dirichlet-to-Neumann map
$\Lambda_{g}^{\rm gr}: \ f\rightarrow
\partial_{\nu}u^{f}|_{\Gamma_{0}}$, where $\nu$ is the outward
normal.

The second problem under consideration is
\begin{equation}
\label{problem is} \Delta_{g} v=0\,\,\, \text{ in }\,\, {\rm
int}M,\qquad v=h \,\,\text{ on }\, \Gamma_{0}, \qquad
\partial_{\nu}v=0\,\, \text{ on }\, \tilde{\Gamma};
\end{equation}
let $v^{h}$ be the solution for a smooth $h$. The DN map
associated with (\ref{problem is}) is $\Lambda_{g}^{\rm is}: \
h\rightarrow
\partial_{\nu}v^{h}|_{\Gamma_{0}}$.
\smallskip

\noindent$\bullet$\,\,\,Assume that one of the DN maps
$\Lambda_{g}^{\rm gr}$ or $\Lambda_{g}^{\rm is}$ is known; the
Electric Impedance Tomography problem is to construct a relevant
copy $(M',g',\Gamma_{0}')$ of $(M,g,\Gamma_{0})$, which is
provides $\Gamma_{0}'=\Gamma_{0}$ and $\Lambda_{g'}=\Lambda_{g}$.
Note that to construct such a copy is the only relevant
understanding of `to solve the EIT problem for the {\it unknown}
manifold' \cite{BCald,B Sobolev Geom Rings,B UMN}.

To solve the EIT problem, we apply the algebraic approach, which
was proposed in \cite{BCald} for surfaces with one-component
boundary. As well as in \cite{BCald}, the main tool for solving
the problem is an algebra of the boundary values of functions
holomorphic in a manifold. However, in the present paper, the
manifold is not $M$ but ${\mathbb M}$, which is obtained by gluing
two examples of $M$ along the `inner' boundary $\tilde{\Gamma}$.
We show that the corresponding algebra is determined by
$\Lambda_{g}^{\rm gr}$ or $\Lambda_{g}^{\rm is}$ up to the
isometric isomorphism. Its Gelfand spectrum (the set of
characters) plays the role of the material for constructing a
relevant copy $(M',g',\Gamma_{0}')$ of $(M,g,\Gamma_{0})$. This
copy is conformally equivalent to the original and thus solves the
problem.

\subsubsection*{Preliminaries}
\noindent$\bullet$\,\,\, Recall that $\Lambda_{g}^{\rm gr}$ and
$\Lambda_{g}^{\rm is}$ are the positive selfadjoint 1-st order
pseudo-diffe\-ren\-ti\-al operators in the {\it real} space
$L_2(\Gamma_{0})$ defined on ${\rm Dom\,}\Lambda_{g}^{\rm gr}={\rm
Dom\,}\Lambda_{g}^{\rm is}=H^1(\Gamma_{0})$. The metric on
$\Gamma_{0}$ can be determined from the principal symbol of
$\Lambda_{g}^{\rm gr}$ or $\Lambda_{g}^{\rm is}$ \cite{Taylor}.
Therefore, in the subsequent we assume that length element $ds$ on
$\Gamma_{0}$ is known and a continuous tangent field of unit
vectors $\gamma$ on $\Gamma_0$ is chosen. For functions on the
boundary, by $\partial_\gamma f$ we denote the derivative with
respect to the length $s$ in direction $\gamma$. Define
 $$
\dot{L}_{2}(\Gamma_{0}):=\linebreak\{f\in
L_2(\Gamma_{0})\,|\,\,\int_{\Gamma_{0}} f\,ds=0\}, \quad
\dot{C}^{\infty}(\Gamma_{0})=C^{\infty}(\Gamma_{0})\cap\dot{L}_{2}(\Gamma_{0}).
 $$
It is easy to see that ${\rm Ker\,}\Lambda_{g}^{\rm is}=\{{\rm
const}\}$ and ${\rm Ran\,}\Lambda_{g}^{\rm
is}=\dot{L}_{2}(\Gamma_{0})$. If $\tilde{\Gamma}\ne\varnothing$,
then the relations ${\rm Ker\,}\Lambda_{g}^{\rm gr}=\{0\}$ and
${\rm Ran\,}\Lambda_{g}^{\rm gr}=\dot L_{2}(\Gamma_{0})$ hold.
Note that $\Lambda_{g}^{\rm gr}=\Lambda_{g}^{\rm is}$ for
$\tilde{\Gamma}=\varnothing$.

\noindent{$\bullet$}\,\,\ There are two orientations on $M$ and
each of them corresponds to a continuous family of rotations $M\ni
x\rightarrow\Phi(x)\in {\rm End\,}T_{x}M$ such that
$$
\langle\Phi(x)\,a,\Phi(x)\,b\rangle=\langle a,b\rangle, \quad
\langle\Phi(x)\,a,a\rangle=0, \qquad a,b\in T_{x}M, \,\, x\in M.
$$
Note that $\Phi(x)^{2}=-{\rm id}$, $x\in M$. Let us fix the
orientation  (and, hence, $\Phi$) by the rule
$\Phi(\gamma)\nu=\gamma$. Then such a rule determines the tangent
field of unit vectors $\gamma$ on each component $\Gamma_{j}$,
$j=1,\dots,m$.

The (real) functions $u,u^{\dag}\in C(M)\cap C^\infty({\rm
int\,}M)$ are {\it conjugate} by Cauchy-Riemann if they satisfy
 \begin{equation}\label{def CR}
\nabla u^{\dag}=\Phi\nabla u\qquad\text{in}\,\,\,\,{\rm int}M.
 \end{equation}
In such a case, the complex-valued function $w=u+iu^{\dag}$ is
{\it holomorphic} in ${\rm int}M$. Obviously, $u$ and $u^{\dag}$
are harmonic in ${\rm int}M$ and $u+c_{1}$, $u^{\dag}+c_{2}$ are
conjugate for any $c_{1},c_{2}\in\mathbb{R}$. The class of
functions $u\in C^{\infty}(M)$ which have a conjugate $u^{\dag}\in
C^{\infty}(M)$ is infinite-dimensional.

\noindent{$\bullet$}\,\,\ Introduce the integration $J: \
\dot{L}_{2}(\Gamma_{0})\rightarrow\dot{L}_{2}(\Gamma_{0})$ by
$\partial_{\gamma}J={\rm id}$. Note that $J\partial_{\gamma}={\rm
id}$ on $C^\infty(\Gamma_0)$.

Suppose that $\tilde{\Gamma}=\varnothing$, whereas $u,u^{\dag}\in
C^{\infty}(M)$ are conjugate. Then (\ref{def CR}) implies
$\partial_{\nu}u=\partial_{\gamma}u^{\dag}$ and
$\partial_{\gamma}u=-\partial_{\nu}u^{\dag}$. Since $u+c_{1}$,
$u^{\dag}+c_{2}$ are also conjugate, we can chose $u$ and
$u^{\dag}$ in such a way that their traces $f:=u|_{\Gamma_{0}}$
and $p=u^{\dag}|_{\Gamma_{0}}$ belong to
$\dot{L}_{2}(\Gamma_{0})$. Then, denoting by $\Lambda_{g}$ the
operator $\Lambda_{g}^{\rm gr}=\Lambda_{g}^{\rm is}$, one has
$\Lambda_{g}f=\partial_{\gamma}p$,
$\partial_{\gamma}f=-\Lambda_{g}p$ and obtains
 $$
[I+(\Lambda_{g}J)^{2}]\partial_{\gamma}f=\partial_{\gamma}f+\Lambda_{g}J\Lambda_g
f=\partial_{\gamma}f+\Lambda_{g}p=0,
 $$
and hence ${\rm Ker\,}[I+(\Lambda_{g}J)^{2}]\not=\{0\}$. Moreover,
${\rm dim\,}{\rm Ker\,}[I+(\Lambda_{g}J)^{2}]=\infty$ holds: see
\cite{BCald}.

Suppose that $\tilde{\Gamma}\ne\varnothing$ and
$[I+(\Lambda_{g}J)^{2}]k=0$, where $k\in
\dot{C}^{\infty}(\Gamma_{0})$ and $\Lambda_{g}=\Lambda_{g}^{\rm
gr}$ or $\Lambda_{g}^{\rm is}$. Then for $u=u^{Jk}$ and
$u'=u^{(\Lambda_{g}J)^{2}k}$ one has
$\partial_{\nu}u=\partial_{\gamma}u'$ and
$\partial_{\gamma}u=-\partial_{\nu}u'$ on $\Gamma_{0}$. Thus,
$\nabla u'=\Phi\nabla u$ holds on $\Gamma_{0}$. By the Poincare
Theorem, for any neighborhood $U\subset M$ homeomorphic to a disk
in $\mathbb R^2$ and such that $\partial U\cap\Gamma_{0}$ contains
a segment $\Gamma'$ of positive length, there exists, a conjugate
function $u''$ such that $\nabla u''=\Phi\nabla u$ in $U$.
Therefore,
$\partial_{\gamma}u'=\partial_{\nu}u=\partial_{\gamma}u''$ and
$-\partial_{\nu}u'=\partial_{\gamma}u=-\partial_{\nu}u''$ on
$\Gamma'$. Due to the uniqueness of the solution to the Cauchy
problem for the second order elliptic equations, one has
$u'=u''+{\rm const}$ on $U$, whence $\nabla u'=\Phi\nabla u$ in
$U$. Since $U$ is arbitrary, the function $w:=u+iu'$ is
holomorphic in ${\rm int}M$. If $\Lambda_{g}=\Lambda_{g}^{\rm
gr}$, then $w=0$ on $\tilde{\Gamma}$ and, due to analyticity,
$w=0$ on $M$. If $\Lambda_{g}=\Lambda_{g}^{\rm is}$, then
$\partial_{\gamma}u'=\partial_{\nu}u=0$ and
$\partial_{\gamma}u=-\partial_{\nu}u'=0$ on $\tilde{\Gamma}$.
Therefore $w={\rm const}$ on $\tilde\Gamma$ and, by analyticity,
$w={\rm const}$ on $M$. In both cases, $Jk=u|_{\Gamma_{0}}={\rm
const}$ and hence $k=0$.

So, for a given operator $\Lambda_{g}=\Lambda_{g}^{\rm gr}$ or
$\Lambda_{g}^{\rm is}$, only the following three alternative cases
are realizable:
\begin{itemize}
\item[-] If the equation $[I+(\Lambda_{g}J)^{2}]k=0$ has
infinitely many (linearly independent) smooth solutions, then
$\tilde{\Gamma}=\varnothing$ and $\Lambda_{g}=\Lambda_{g}^{\rm
gr}=\Lambda_{g}^{\rm is}$. \item[-] If $\Lambda_{g}1\ne 0$ and
$[I+(\Lambda_{g}J)^{2}]k\ne 0$ for any nonzero
$k\in\dot{C}^{\infty}(\Gamma_{0})$, then
$\tilde{\Gamma}\ne\varnothing$ and $\Lambda_{g}=\Lambda_{g}^{\rm
gr}$. \item[-] If $\Lambda_{g}1=0$ and $[I+(\Lambda_{g}J)^{2}]k\ne
0$ for any nonzero $k\in\dot{C}^{\infty}(\Gamma_{0})$, then
$\tilde{\Gamma}\ne\varnothing$ and $\Lambda_{g}=\Lambda_{g}^{\rm
is}$.
\end{itemize}

\subsubsection*{Algebras on $M$}
\noindent{$\bullet$}\,\,\ From now on, we suppose that
$\tilde{\Gamma}\ne\varnothing$, referring the reader to
\cite{BCald} for the case $\tilde{\Gamma}=\varnothing$. Let
\begin{equation}\label{hermelll}
\mathfrak{A}_{*}(M):=\{w=v+iu \ |\,\, u,v\in C(M), \ \nabla
u=\Phi\nabla v \text{ in } {\rm int}M, \ u=0 \text{ on }
\tilde{\Gamma}\}
\end{equation}
be the set of holomorphic continuous functions with real traces on
$\tilde{\Gamma}$; we also denote
$\mathfrak{A}_{*}^{\infty}(M):=\mathfrak{A}_{*}(M)\cap
C^{\infty}(M;\mathbb{C})$. Note that ${\rm
clos\,}\mathfrak{A}_{*}^{\infty}(M)=\mathfrak{A}_{*}(M)$ (the
closure in $C^{\infty}(M;\mathbb{C})$).

Also, note that $\mathfrak{A}_{*}(M)$ is not a {\it complex}
algebra because $i\,\mathfrak{A}_{*}(M)\not=\mathfrak{A}_{*}(M)$.
At the same time, it is easy to check that
 \begin{equation}\label{def mathbb A(M)}
\mathbb{A}(M):=\mathfrak{A}_{*}(M)+i\,\mathfrak{A}_{*}(M)=\{a+ib\,|\,\,a,b\in\mathfrak{A}_{*}(M)\}
 \end{equation}
{\it is an algebra} in $C(M;\mathbb{C})$. Its smooth subalgebra
$\mathbb{A}^{\infty}(M):=\mathbb{A}(M)\cap
C^{\infty}(M;\mathbb{C})$ is dense in $\mathbb A(M)$.

For any element $\zeta\in\mathbb{A}(M)$, the representation
$\zeta=w_{1}+iw_{2}$ with $w_{1},w_{2}\in\mathfrak{A}_{*}(M)$ is
unique. Indeed, if $w_{1}+iw_{2}=w_{3}+iw_{4}$ with
$w_{k}\in\mathfrak{A}_{*}(M)$, then
$\tilde{w}:=w_{1}-w_{3}=i(w_{4}-w_{2})$ is an element of
$\mathfrak{A}_{*}(M)\cap i\mathfrak{A}_{*}(M)$. Hence,
$\tilde{w}=0$ on $\tilde{\Gamma}$ and, by analyticity,
$\tilde{w}=0$ on $M$. Thus, $w_{1}=w_{3}$ and $w_{2}=w_{4}$.

Also, one can simply verify that the map
\begin{equation}
\label{invdef}
w_{1}+iw_{2}\rightarrow(w_{1}+iw_{2})^{*}:=w_{1}-iw_{2}, \qquad w_{1},w_{2}\in\mathfrak{A}_{*}(M)
\end{equation}
obeys $(\zeta_{1}\zeta_{2})^{*}=\zeta_{2}^{*}\zeta_{1}^{*}$ and,
hence, is an involution on the algebra $\mathbb{A}(M)$.

For $\zeta\in\mathbb{A}(M)$, we put
\begin{equation}
\label{normdef}
|||\zeta|||:=\max\,\{\parallel\zeta\parallel_{C(M;\mathbb{C})},\parallel\zeta^{*}\parallel_{C(M;\mathbb{C})}\}
\end{equation}
and see that $|||\zeta|||=|||\zeta^{*}|||$ holds. Then, easily
checking the property
$|||\zeta\eta|||\leqslant|||\zeta|||\,|||\eta|||$, we conclude
that $\{\mathbb{A}(M),|||\cdot |||,*\}$ is an involutive Banach
algebra.
\smallskip

\noindent{$\bullet$}\,\,\, In accordance with the uniqueness of
analytic continuation, the trace map
\begin{equation}
\label{tr 1} \mathbb{A}(M)\ni \zeta\stackrel{{\rm
Tr}_{\,\Gamma_{0}}}{\rightarrow}\zeta|_{\Gamma_{0}}\,\in
C(\Gamma_0;\mathbb C)
\end{equation}
is an isomorphism of the algebras $\mathbb{A}(M)$ and ${\rm
Tr}_{\,\Gamma_{0}}\mathbb{A}(M)$. This isomorphism determines the
involution of the trace algebra
\begin{equation}
\label{invdef+}
\eta_{1}+i\eta_{2}\rightarrow(\eta_{1}+i\eta_{2})^{*}:=\eta_{1}-i\eta_{2},
\qquad \eta_{1},\eta_{2}\in{\rm Tr}_{\Gamma_0}\mathfrak{A}_{*}(M).
\end{equation}
In what follows, we prove that ${\rm Tr}_{\,\Gamma_{0}}$ is an
isometry between $\mathbb{A}(M)$ and ${\rm
Tr}_{\,\Gamma_{0}}\mathbb{A}(M)$, where the norm in ${\rm
Tr}_{\,\Gamma_{0}}\mathbb{A}(M)$ is defined by (\ref{normdef})
with $\Gamma_{0}$ instead of $M$.

Now, let us show that the algebra ${\rm
Tr}_{\Gamma_{0}}\mathbb{A}(M)\subset C(\Gamma_0;\mathbb C)$ is
determined by the (known) DN map $\Lambda_{g}^{\rm gr}$ or
$\Lambda_{g}^{\rm is}$.
\begin{Lemma}
\label{criteria grounded prop}
For $f\in C^{\infty}(\Gamma_{0})$ the following conditions are equivalent:
\begin{enumerate}
\item $u^{f}\in \Re(i\,\mathfrak{A}_{*}^{\infty}(M))$; \item
\begin{equation} \label{criteria grounded} 2\Lambda_{g}^{\rm gr}(f\cdot J\Lambda_{g}^{\rm gr}f)=\partial_{\gamma}[(J\Lambda_{g}^{\rm gr}f)^{2}-f^{2}].
\end{equation}
\end{enumerate}
\end{Lemma}
\begin{proof}

${\bf 1.\Rightarrow 2.}$\,\,\,Let $w=u^{f}+iv$ belong to
$i\,\mathfrak{A}_{*}^{\infty}(M)$; then $u^{f}=0$ on
$\tilde{\Gamma}$. Denote $h=v|_{\Gamma_{0}}$. The function $w^{2}$
is holomorphic in $M$ and, hence, $\Im w^{2}=2u^{f}v$ is harmonic
in $M$. Moreover, $2u^{f}v=0$ on $\tilde{\Gamma}$. Thus,
$2u^{f}v=u^{2fh}$ and
\begin{equation}
\label{squaring trick 1} 2\Lambda_{g}^{\rm
gr}(fh)=2\partial_{\nu}(u^{f}v)=2(f\partial_{\nu}v+h\Lambda_{g}^{\rm
gr}f).
\end{equation}
In addition, the Cauchy-Riemann conditions $\Phi\nabla
u^{f}=\nabla v$ in $M$ imply
\begin{equation}
\label{CR1} \Lambda_{g}^{\rm
gr}f=\partial_{\nu}u^{f}=\partial_{\gamma}h, \quad
\partial_{\nu}v=-\partial_{\gamma}f \quad \text{ on } \Gamma_{0}.
\end{equation}
Integrating the first equality of (\ref{CR1}), one gets
$h=J\Lambda_{g}^{\rm gr}f+c$ with a constant $c$. Note that,
replacing $v$ by $v-c$, one preserves the inclusion
$w\in\mathfrak{A}_{*}^{\infty}(M)$. Thus, we can choose $v$ in
such a way that
\begin{equation}
\label{tr1} h=J\Lambda_{g}^{\rm gr}f
\end{equation}
Substituting (\ref{tr1}) and the second equality of (\ref{CR1}) in
(\ref{squaring trick 1}), one obtains
$$2\Lambda_{g}^{\rm gr}(f\cdot J\Lambda_{g}^{\rm gr}f)=2(\Lambda_{g}^{\rm gr}f\cdot J\Lambda_{g}^{\rm gr}f-f\partial_{\gamma}f)=\partial_{\gamma}[(J\Lambda_{g}^{\rm gr}f)^{2}-f^{2}].$$

${\bf 2.\Rightarrow 1.}$\,\,\, Choose a segment
$\Gamma'\subset\Gamma_{0}$ of positive length and an arbitrary
neighbourhood $U\subset M$ diffeomorphic to a disc in
$\mathbb{R}^{2}$ provided $\partial U\supset\Gamma'$. By the
Poincare Theorem, there is a harmonic in $U$ function $v_{U}$
obeying $\Phi\nabla {\rm u}^{\rm f}=\nabla v_{U}$ in $U$; such a
function is defined up to a constant. In particular,
\begin{equation}
\label{CR2}
\partial_{\gamma}v_{U}=\partial_{\nu}u^{f}=\Lambda_{g}^{\rm gr}f, \quad \partial_{\nu}v_{U}=-\partial_{\gamma}f \quad \text{ on } \Gamma'.
\end{equation}
The first equality of (\ref{CR2}) shows that $v_{U}$ can be chosen
in such a way that $v_{U}=J\Lambda_{g}^{\rm gr}f$ on $\Gamma'$.
The latter, along with the second equality of (\ref{CR2}), implies
\begin{equation}
\label{cdata 1}
\partial_{\nu}(u^{f}v_{U})=u^{f}\partial_{\nu}v_{U}+v_{U}\partial_{\nu}u^{f}=J\Lambda_{g}^{\rm gr}f\cdot\Lambda_{g}^{\rm gr}f-f\partial_{\gamma}f=\frac{1}{2}\,\partial_{\gamma}[(J\Lambda_{g}^{\rm gr}f)^{2}-f^{2}].
\end{equation}
Comparing (\ref{cdata 1}) with condition $2.$, we obtain
$$\partial_{\nu}(u^{f}v_{U})=\Lambda_{g}^{\rm gr}(f\cdot J\Lambda_{g}^{\rm gr}f)=\partial_{\nu}u^{f\cdot J\Lambda_{g}^{\rm gr}f}.$$
Since $w_{U}:=u^{f}+iv_{U}$ is holomorphic in $U$, the function
$w^{2}_{U}$ is also holomorphic in $U$ and, hence,
$\frac{1}{2}\,\Im w^{2}_{U}=u^{f}v_{U}$ is harmonic in $U$. So,
$u^{f}v_{U}$ and $u^{f\cdot J\Lambda_{g}^{\rm gr}f}$ are both
harmonic in $U$ and have the same Cauchy data on $\Gamma'$. Thus,
$u^{f}v_{U}=u^{f\cdot J\Lambda_{g}^{\rm gr}f}$ in $U$ due to the
uniqueness of the solution to the Cauchy problem for the second
order elliptic equations. Therefore, outside the (possible) zeros
of $u^f$ in $U$, the function
$$w:=u^{f}+i\,\frac{u^{f\cdot J\Lambda_{g}^{\rm gr}f}}{u^{f}}$$
coincides with the function $w_{U}$ holomorphic in $U$. Since $U$
is arbitrary, $w$ is holomorphic in $M$. Moreover, $\Re w=u^{f}=0$
on $\tilde{\Gamma}$ and $w\in i\,\mathfrak{A}_{*}^{\infty}(M)$.
The latter means that $u^{f}\in
\Re(i\,\mathfrak{A}_{*}^{\infty}(M))$ holds.
\end{proof}
As a corollary of Lemma \ref{criteria grounded prop}, we represent
\begin{equation}
\label{hol func via DN gr} {\rm
Tr}_{\Gamma_{0}}\mathfrak{A}_{*}^{\infty}(M)=\{\eta=J\Lambda_{g}^{\rm
gr}f-if+c \ | \ f\in C^{\infty}(\Gamma_{0}) \text{ obeys
(\ref{criteria grounded})}, \ c\in\mathbb{R}\}
\end{equation}
and conclude that the algebra ${\rm Tr}_{\Gamma_{0}}\mathbb{A}(M)$
is determined by the DN map $\Lambda_{g}^{\rm gr}$.
\smallskip

\noindent{$\bullet$}\,\,\,Now, let us show that the set ${\rm
Tr}_{\,\Gamma_{0}}\mathfrak{A}_{*}^{\infty}(M)$ is also determined
by the DN map $\Lambda_{g}^{\rm is}$.
\begin{Lemma}
\label{criteria isolated prop}
For $h\in C^{\infty}(\Gamma_{0})$ the following conditions are equivalent:
\begin{enumerate}
\item $v^{h}\in \Re(\mathfrak{A}_{*}^{\infty}(M))$;
\item there is a (real) number $c_{h}$ such that
\begin{equation}
\label{criteria isolated} \frac{1}{2}\Lambda_{g}^{\rm
is}[h^{2}-(J\Lambda_{g}^{\rm is}h)^{2}]-h\Lambda_{g}^{\rm
is}h-J\Lambda_{g}^{\rm
is}h\cdot\partial_{\gamma}h=c_{h}(\partial_{\gamma}+\Lambda_{g}^{\rm
is}J\Lambda_{g}^{\rm is})h\quad \text{ on }\, \Gamma_{0}.
\end{equation}
\end{enumerate}
\end{Lemma}
\begin{proof}
Obviously, the statement is valid for $h={\rm const}$. In what
follows we deal with $h\ne {\rm const}$.

${\bf 1.\Rightarrow 2.}$\,\,\,Let the function $w=v^{h}+iu$ belong
to $\mathfrak{A}^{\infty}_{*}(M)$; then $u=0$ on $\tilde{\Gamma}$.
Denote $f=u|_{\Gamma_{0}}$. The function $w^{2}$ is holomorphic in
$M$ and, hence, $\Re w^{2}=(v^{h})^{2}-u^{2}$ is harmonic in $M$.
Moreover, $\partial_{\nu}v^{h}=\partial_{\gamma}u=0$ on
$\tilde{\Gamma}$, whence we have
$$\partial_{\nu}[(v^{h})^{2}-u^{2}]=2(v^{h}\partial_{\nu}v^{h}-u\partial_{\nu}u)=0\qquad \text{ on }\,\, \tilde{\Gamma}.$$
Thus, $(v^{h})^{2}-u^{2}=v^{h^{2}-f^{2}}$ holds and we have
\begin{equation}
\label{squaring trick 2} \Lambda_{g}^{\rm
is}(h^{2}-f^{2})=\partial_{\nu}[(v^{h})^{2}-u^{2}]=2(h\Lambda_{g}^{\rm
is}h-f\partial_{\nu}u)\quad \text{ on }\,\, \Gamma_{0}.
\end{equation}
In addition, the Cauchy-Riemann conditions $\Phi\nabla v^{h}=\nabla u$ in $M$ imply
\begin{equation}
\label{CR3}
\partial_{\gamma}f=\Lambda_{g}^{\rm is}h, \quad \partial_{\nu}u=-\partial_{\gamma}h \quad \text{ on }\,\, \Gamma_{0}.
\end{equation}
The first equality of (\ref{CR3}) means that $f=J\Lambda_{g}^{\rm
is}h+c_{h}$ with some $c_{h}\in\mathbb{R}$. Taking into account
this and the second equality of (\ref{CR3}), we rewrite
(\ref{squaring trick 2}) as
 $$
\Lambda_{g}^{\rm is}(h^{2}-(J\Lambda_{g}^{\rm
is}h+c_{h})^{2})=2[\,h\Lambda_{g}^{\rm is}h+(J\Lambda_{g}^{\rm
is}h+c_{h})\partial_{\gamma}h\,].
 $$
Since $\Lambda_{g}^{\rm is}(c_{f}^{2})=0$, the latter implies
(\ref{criteria isolated}).

${\bf 2.\Rightarrow1.}$\,\,\, Since $\Lambda_{g}^{\rm
is}(c_{h}^{2})=0$, (\ref{criteria isolated}) can be rewritten as
\begin{equation}
\label{squaring trick 3} \Lambda_{g}^{\rm
is}(h^{2}-f^{2})=2(h\Lambda_{g}^{\rm is}h+f\partial_{\gamma}h),
\end{equation}
where $f:=J\Lambda_{g}^{\rm is}h+c_{h}$. Choose a segment
$\Gamma'\subset\Gamma_{0}$ of positive length and an arbitrary
neighbourhood $U\subset M$ diffeomorphic to a disc in
$\mathbb{R}^{2}$ and such that $\partial U\supset\Gamma'$. By the
Poincare Theorem, there is a harmonic in $U$ function $u_{U}$
provided $\Phi\nabla v^{h}=\nabla u_{U}$ in $U$; such a function
is defined up to a constant. In particular,
\begin{equation}
\label{CR4}
\partial_{\gamma}u_{U}=\partial_{\nu}v^{h}=\Lambda_{g}^{\rm is}h=\partial_{\gamma}f, \quad -\partial_{\gamma}h=\partial_{\nu}u_{U}\quad \text{ on } \Gamma'.
\end{equation}
In view of the first equality of (\ref{CR4}), $u_{U}$ can be
chosen is such a way that $u_{U}=f$ on $\Gamma'$. The second
equality of (\ref{CR4}) implies
\begin{equation}
\label{squaring trick 4}
\partial_{\nu}[(v^{h})^{2}-u_{U}^{2}]=2(h\Lambda_{g}^{\rm is}h-f\partial_{\nu}u_{U})=2(h\Lambda_{g}^{\rm is}h+f\partial_{\gamma}h) \text{ on } \Gamma'.
\end{equation}
Comparing (\ref{squaring trick 4}) with (\ref{squaring trick 3}), we obtain
$$\partial_{\nu}((v^{h})^{2}-u_{U}^{2})=\Lambda_{g}^{\rm is}(h^{2}-f^{2})=\partial_{\nu}v^{h^{2}-f^{2}} \text{ on } \Gamma'.$$
Since $v^{h}+iu_{U}:=w_{U}$ is holomorphic in $U$, the function
$w_{U}^{2}$ is also holomorphic in $U$ and, hence, $\Re
w_{U}=(v^{h})^{2}-u_{U}^{2}$ is harmonic in $U$. So,
$(v^{h})^{2}-u_{U}^{2}$ and $v^{h^{2}-f^{2}}$ satisfy the
Laplace-Beltrami equation in $U$ and have the same Cauchy data on
$\Gamma'$. Thus, $(v^{h})^{2}-u_{U}^{2}=v^{h^{2}-f^{2}}$ in $U$
due to the uniqueness of the solution to the Cauchy problem for
the second order elliptic equations.

Introduce the function $u$ on $M$ by the rule $u(x):=u_{U}(x)$,
where $U$ be a neighbourhood diffeomorphic to a disc in
$\mathbb{R}^{2}$ and such that $x\in U$ and $\partial
U\supset\Gamma'$. Let us check that such a definition does not
depend on the choice of $U$. Suppose that $U_{1},U_{2}$ are
neighbourhoods diffeomorphic to a disc in $\mathbb{R}^{2}$ and
such that $\partial U_{1,2}\supset\Gamma'$. Then, as is shown
above, $u_{U_{1}}^{2}=(v^{h})^{2}-v^{h^{2}-f^{2}}=u_{U_{2}}^{2}$
and, hence, $u_{U_{1}}=su_{U_{2}}$ on $U_{1}\cap U_{2}$, where
$s=1$ or $s=-1$. Since, the Cauchy-Riemann conditions
$$\nabla u_{U_{1}}=\Phi\nabla v^{h}=\nabla u_{U_{2}}=s\nabla u_{U_{1}}$$
hold on $U_{1}\cap U_{2}$ and $h$, $v^{h}$ are non-constant, the
last formula yields $s=1$. Thus, $u_{U_{1}}=u_{U_{2}}$ on
$U_{1}\cap U_{2}$. So, $u$ is well defined and the function
$w:=v^{h}+iu$ is holomorphic in $M$. From the Cauchy-Riemann
conditions $\Phi\nabla v^{h}=\nabla u$ it follows that
$\partial_{\gamma}u=0$ on $\tilde{\Gamma}$. The latter implies
that $\partial_{\nu}u$ is not identical zero on $\Gamma_{j}$ for
any $j=1,\dots,m$. Indeed, the opposite means that $\nabla u=0$ on
$\Gamma_{j}$ and then $u={\rm const}$ on $M$ due to the uniqueness
of the solution to the Cauchy problem for the second order
elliptic equations. Then $\nabla v^{h}=-\Phi\nabla u=0$ in $M$,
which contradicts the assumption $h\ne {\rm const}$ on
$\Gamma_{0}$. As a result, the equality
$$0=\partial_{\nu}v^{h^{2}-f^{2}}=\partial_{\nu}((v^{h})^{2}-u^{2})=2(v^{h}\partial_{\nu}v^{h}-u\partial_{\nu}u)=2u\partial_{\nu}u \text{ on } \tilde{\Gamma}$$
yields $u|_{\tilde{\Gamma}}=0$.
\end{proof}
Lemma \ref{criteria isolated prop} leads to the representation
\begin{equation}
\label{hol func via DN is} {\rm
Tr}_{\Gamma_{0}}\mathfrak{A}_{*}^{\infty}(M)=\{\eta=h+i(J\Lambda_{g}^{\rm
is}h+c_{h})H_{h} \ | \ h\in C^{\infty}(\Gamma_{0}) \text{ and }
c_{h}\in\mathbb{R} \text{ obey (\ref{criteria isolated})}\},
\end{equation}
where $H_{\rm const}=0$ and $H_{f}=1$ for any nonconstant $f$.
From (\ref{hol func via DN gr}) and (\ref{hol func via DN is}) it
follows that
\begin{equation}
\label{tr alg via DNs}
\begin{split}
{\rm Tr}_{\Gamma_{0}}\mathbb{A}^{\infty}(M)&=\{\eta=f_{1}+J\Lambda_{g}^{\rm gr}f_{2}+i(J\Lambda_{g}^{\rm gr}f_{1}-f_{2})+c \ | \ \\
&\qquad\qquad f_{1},f_{2}\in C^{\infty}(\Gamma_{0}) \text{ obey (\ref{criteria grounded})}, \ c\in\mathbb{C}\}=\\
&=\{\eta=h_{1}-J\Lambda_{g}^{\rm is}h_{2}-c_{2}+i(J\Lambda_{g}^{\rm is}h_{1}+c_{1}+h_{2}) \ | \ \\
&h_{k}\in C^{\infty}(\Gamma_{0}), \ c_{k}\in\mathbb{R}; \ (h,c_{h}):=(h_{k},c_{k}) \text{ obey (\ref{criteria isolated})}\}.
\end{split}
\end{equation}
Since ${\rm Tr}_{\Gamma_{0}}\mathbb{A}^\infty(M)$ is dense in
${\rm Tr}_{\Gamma_{0}}\mathbb{A}(M)$, we arrive at the following
important fact.
 \begin{prop}\label{Prop important}
Each of the operators $\Lambda_{g}^{\rm gr}$ and $\Lambda_{g}^{\rm
is}$ does determine the algebra ${\rm
Tr}_{\,\Gamma_{0}}\mathbb{A}(M)$ via (\ref{tr alg via DNs}).
 \end{prop}

\subsubsection*{Manifold $\mathbb M$}

\noindent{$\bullet$}\,\,\,Take two examples $M_+:=M\times\{+1\}$
and $M_-:=M\times\{-1\}$ of $M$, and factorize $M_+\cup\,M_-$ by
the equivalence
 $$
\begin{cases}
x\times\{+1\} \sim  x\times\{+1\}  &{\rm
if}\,\,x\in M\setminus\tilde\Gamma,\\
x\times\{-1\} \sim  x\times\{-1\}   &{\rm
if}\,\,x\in M\setminus\tilde\Gamma,\\
x\times\{+1\} \sim  x\times\{-1\}   &{\rm if}\,\,x\in\tilde\Gamma
\end{cases}
 $$
(i.e., the examples are glued along $\tilde\Gamma$). Let
$\mathbb{M}$ be the factor space, $x_\pm\in\mathbb M$ the
equivalence classes of $x\times\{\pm1\}$; denote
 $$
\mathbb M_\pm:=\{x_\pm\,|\,\,x\in M_\pm\},\quad \mathbb
M_0:=\{x_0\,|\,\,x\in\tilde\Gamma\}\,,
 $$
where $x_0:=x_+=x_-$ for $x\in\tilde\Gamma$. The factor space is
endowed with the projection $\pi:\mathbb M\to M$ and  the
involution $\tau:\mathbb M\to\mathbb M$ by
 $$
\pi: x_\pm \mapsto x\quad\text{and}\quad\tau: x_\pm\mapsto x_\mp.
 $$
Each $x_0$ is a fixed point of $\tau$, i.e.,
 \begin{equation}\label{fixed point}
\tau(x_0)=x_0,\qquad x_0\in\mathbb M_0
 \end{equation}
holds. By construction, $\{\mathbb{M},\pi\}$ is a covering of $M$.
\smallskip

\noindent{$\bullet$}\,\,\, Here we endow $\mathbb M$ with the
relevant metric and analytic structure.

${\rm a})$\,\,Choose a smooth atlas on $M$ consisting of the
charts of the following two types. The charts $(U_{j},\phi_{j})$
of the first type obey
$\overline{U_{j}}\cap\tilde{\Gamma}=\varnothing$, whereas
$\phi_{j}(U_{j})$ belongs to the half-plane
$\Pi_{+}=\{(x_{1},x_{2}) \ | \ x_{2}>0\}$. The charts
$(U_{j},\phi_{j})$ of the second type are chosen so that
$\phi_{j}(U_{j})\subset\Pi_{+}$ holds and
$\overline{U_{j}}\cap\tilde{\Gamma}$ is a segment $\sigma_{j}$ of
$\tilde{\Gamma}$ such that $\phi_{j}(\sigma_j)$ is a segment
$[a_j,b_j]$ of the axis $x_{2}=0$.

Now, let us construct the smooth atlas on $\mathbb{M}$. For the
chart $(U_{j},\phi_{j})$ of the first type on $M$, one construct
two charts $(\mathbb{U}^{\pm}_{j},\phi^{\pm}_{j})$ on
$\mathbb{M}$, where $\mathbb
U_{j}^{\pm}=\pi^{-1}(U_{j})\cap\mathbb{M}_{\pm}$,
$\phi_{j}^{+}=\phi_{j}\circ\pi$, and
$\phi_{j}^{-}=\kappa\circ\phi_{j}\circ\pi$; here
$\kappa(x_{1},x_{2}):=(x_{1},-x_{2})$. Each chart
$(U_{j},\phi_{j})$ of the second type on $M$ determines the chart
$(\mathbb{U}^{0}_{j},\phi^{0}_{j})$ on $\mathbb{M}$, where
$\mathbb{U}^{0}_{j}=\pi^{-1}(U_{j}\cup\sigma_{j})$ and
\begin{eqnarray*}
\phi_{j}^{0}(x):=\left\{ \begin{array}{ll}
\phi_{j}\circ\pi, & x\in\mathbb{U}^{0}_{j}\cap\mathbb{M}_{+},\\
\kappa\circ\phi_{j}\circ\pi(x), &
x\in\mathbb{U}^{0}_{j}\cap\mathbb{M}_{-}.
\end{array}\right.
\end{eqnarray*}
As is easy to verify, all together the charts
$(\mathbb{U}^{\pm}_{j},\phi^{\pm}_{j})$ and
$(\mathbb{U}^{0}_{j},\phi^{0}_{j})$ make up a smooth atlas on
$\mathbb M$. So, we have the smooth manifold $\mathbb M$ with the
boundary
 $$
\partial\mathbb
M=\Gamma^+_0\cup\Gamma^-_0,\quad\text{where}\quad\Gamma^\pm_0:=\pi^{-1}(\Gamma_0\times\{\pm
1\}).
 $$

${\rm b})$\,\,Next, $\mathbb{M}$ is endowed with the metric ${\rm
g}:=\pi_{*}g$. Such a metric is invariant with respect to the
involution $\tau$, i.e., obeys $\tau_{*}{\rm g}={\rm g}$. The
metric ${\rm g}$ is smooth outside $\mathbb{M}_{0}$. However, on
the whole $\mathbb{M}$ it is only Lipschitz continuous in view of
possible break of smoothness on $\mathbb{M}_{0}$.

Given metric $\rm g$, one defines the continuous field of
rotations $\dot{\Phi}$ (which is a tensor field on $\mathbb{M}$)
such that $\dot{\Phi}|_{\mathbb{M}_{+}}=\pi_{*}\Phi$ and
$\tau_{*}\dot{\Phi}=-\dot{\Phi}$ hold.
\smallskip

${\rm c})$\,\,\,One can endow the manifold $M$ with a {\it
biholomorphic} atlas as follows. Let $x$ be an arbitrary point of
$\mathbb{M}$, and $(\mathbb{U},\psi)$ is a chart obeying $x\in
\mathbb{U}$. By the Vecua theorem (see Chapter 2, \cite{Vekua}),
there exists a chart $(\mathbb{U}_{x},\psi_{x})$, $x\in
\mathbb{U}_{x} \subset \mathbb{U}$ with {\it isothermal
coordinates} $\psi_{x}$ corresponding to the metric ${\rm g}$ and
such that $\psi_{x}\circ\psi^{-1}$ and $\psi\circ\psi_{x}^{-1}$
are continuously differentiable. In these coordinates, the metric
tensor ${\rm g}$ is of the form ${\rm
g}^{ij}=\rho(\cdot)\delta^{ij}$ with a Lipschitz $\rho>0$.
Respectively, assuming $(\mathbb{U},\psi)$ to be properly oriented
by $\dot\Phi$, the matrix $\dot\Phi_k^l$ in the isothermal
coordinates takes the form
$\dot{\Phi}^{1}_{1}=\dot{\Phi}^{2}_{2}=0$,
$\dot{\Phi}^{1}_{2}=-\dot{\Phi}^{2}_{1}=1$.

As a consequence, if ${\rm w}=v+iu$ satisfies $\Delta_{\rm
g}v=\Delta_{\rm g}u=0$ and $\nabla_{\rm g} u=\dot\Phi\nabla_{\rm
g} v$ in $\mathbb U_x$, then the functions $v\circ\psi_x^{-1}$ and
$u\circ\psi_x^{-1}$ turn out to be harmonic and connected via the
Cauchy-Riemann conditions in $\psi_x(\mathbb U_x)\subset\mathbb
R^2$. Thus, ${\rm w}\circ\psi_x^{-1}$ is holomorphic in
$\psi_x(\mathbb U_x)\subset\mathbb C$ in the classical sense.
Respectively, $\rm w$ is said to be {\it holomorphic in} $\mathbb
U_x$.

At last, covering $\mathbb M$ by the properly oriented isothermal
charts $(\mathbb{U}_{x},\psi_{x})$ and easily checking their
compatibility, one obtains the {\it biholomorphic} atlas
consistent with the metric ${\rm g}$.  As a result, one gets the
possibility to speak about holomorphic functions on $\mathbb{M}$.

\subsubsection*{Algebra on $\mathbb M$}

\noindent{$\bullet$}\,\,\,The basic element of our approach is the
algebra of holomorphic functions
 $$
\mathfrak{A}(\mathbb{M}):=\left\{{\rm w}={\rm v}+i{\rm u} \ | \
{\rm u},{\rm v}\in C(\mathbb{M}); \ \nabla {\rm
u}={\dot\Phi}\nabla {\rm v}\, \text{ in } {\rm
int\,}\mathbb{M}\right\}.
 $$
Endowed with the norm $\parallel {\rm
w}\parallel=\sup_{\mathbb{M}}|{\rm w}|$, algebra
$\mathfrak{A}(\mathbb{M})$ is a closed subalgebra of
$C(\mathbb{M};\mathbb{C})$. The set
$\mathfrak{A}^{\infty}(\mathbb{M})=\mathfrak{A}(\mathbb{M})\cap
C^{\infty}(\mathbb{M})$ is dense in $\mathfrak{A}(\mathbb{M})$.

It is easy to see that the map
$${\rm w}\rightarrow{\rm w}^{\star}:=\overline{{\rm w}\circ\tau}$$
is an involution on $\mathfrak{A}(\mathbb{M})$, and
$\mathfrak{A}(\mathbb{M})$ is an involutive Banach algebra.
Defining $$\mathfrak{A}_{\star}(\mathbb{M}):=\{{\rm
w}\in\mathfrak{A}(\mathbb{M})\,|\,\,{\rm w}^\star={\rm w}\}$$ the
set of Hermitian elements of $\mathfrak{A}(\mathbb{M})$, we have
$\mathfrak{A}(\mathbb{M})=\mathfrak{A}_{\star}(\mathbb{M})+i\mathfrak{A}_{\star}(\mathbb{M})$.
Each ${\rm w}\in\mathfrak{A}(\mathbb{M})$ is representated in the
form ${\rm w}={\rm w}_{1}+i{\rm w}_{2}$, where ${\rm w}_{1}=({\rm
w}+{\rm w}^{\star})/2$ and ${\rm w}_{2}=({\rm w}-{\rm
w}^{\star})/2i$ are hermitian.
\smallskip

\noindent{$\bullet$}\,\,\,By (\ref{fixed point}), one has ${\rm
w}^{\star}=\overline{{\rm w}}$ on $\mathbb M_0$, and any function
${\rm w}\in\mathfrak{A}_{\star}(\mathbb{M})$ is {\it real} on
$\mathbb M_0$. Therefore, restricting ${\rm
w}\in\mathfrak{A}_{\star}(\mathbb{M})$ on $\mathbb M_+$, one gets
${\rm w}|_{\mathbb M_+}=w\circ \pi$ with $w\in\mathfrak{A}_{*}(M)$
(see (\ref{hermelll}) for definition). The converse is also true:
in view of $\tau_{*}{\rm g}={\rm g}$, any function of the form
$w\circ\pi$ with $w\in\mathfrak{A}_{*}(M)$, which is given on
$\mathbb M_+$, admits a (unique) holomorphic continuation ${\rm
w}$ to $\mathbb{M}$ and ${\rm w}(\tau(x))=\overline{{\rm w}(x)}$
holds for all $x\in \mathbb M$. Thus, the algebra $\mathbb{A}(M)$
defined in (\ref{def mathbb A(M)}) and the algebra
$\mathfrak{A}(\mathbb{M})$ are related via the restriction:
 $$
\mathfrak{A}(\mathbb{M})\big|_{\mathbb M_+}:=\{{\rm w}|_{\mathbb
M_+}\,|\,\,{\rm
w}\in\mathfrak{A}(\mathbb{M})\}\,=\,\{w\circ\pi\,|\,\,w\in\mathbb{A}({M})\}=:\mathbb{A}({M})\circ\pi\,.
 $$
In the meantime, by definition (\ref{normdef}), one has
\begin{align*}
&\parallel {\rm w}_{1}+i{\rm w}_{2}\parallel_{C(\mathbb M;\mathbb
C)}=
\max\{\parallel{\rm w}_{1}+i{\rm w}_{2}\parallel_{C(\mathbb M;\mathbb{C})},\parallel{\rm w}_{1}\circ\tau+i{\rm w}_{2}\circ\tau\parallel_{C(\mathbb M;\mathbb{C})}\}=\\
&=\max\{\parallel{\rm w}_{1}+i{\rm w}_{2}\parallel_{C(\mathbb
M;\mathbb{C})},\parallel\overline{{\rm w}_{1}-i{\rm
w}_{2}}\parallel_{C(\mathbb M;\mathbb{C})}\}=|||w_{1}+iw_{2}|||,
\end{align*}
where ${\rm w}_{1},{\rm w}_{2}$ are elements of
$\mathfrak{A}_{\star}(M)$, whereas $w_{1},w_{2}$ are their
`restrictions':  ${\rm w}_{1,2}|_{_{\mathbb
M_+}}=w_{1,2}\circ\pi$. Similarly, (\ref{invdef}) implies $({\rm
w}_{1}+i{\rm w}_{2})^{\star}|_{\mathbb M}=[w_{1}-iw_{2}]\circ\pi$.
Thus, the restriction provides an isometric isomorphism of the
involutive Banach algebra $\mathfrak{A}(\mathbb{M})$ onto the
involutive Banach algebra $\mathbb{A}(M)$.
\smallskip

\noindent{$\bullet$}\,\,\,In accordance with the maximal principle
for holomorphic functions, each ${\rm
w}\in\mathfrak{A}(\mathbb{M})$ obeys $\sup_{\mathbb{M}}|{\rm
w}|=\sup_{\partial\mathbb{M}}|{\rm w}|$. Therefore, the map
$$\mathfrak{A}(\mathbb{M})\ni{\rm w}\stackrel{\rm Tr}{\longrightarrow}\eta:=
{\rm w}|_{\partial\mathbb{M}}\in
C(\partial\mathbb{M};\mathbb{C})$$ is an isometry on its image. In
the meantime, obviously, ${\rm Tr}$ is a homomorphism of algebras.
Thus, ${\rm Tr\,}{\mathfrak A}({\mathbb M})$ is a closed
subalgebra of $C(\partial\mathbb{M};\mathbb{C})$, which is
isometrically isomorphic to $\mathfrak{A}(\mathbb{M})$ via the map
${\rm Tr}$. The involution on $\mathfrak{A}(\mathbb{M})$ induces
the corresponding involution
$\eta\rightarrow\eta^{\star}:=\overline{\eta\circ\tau}$ on ${\rm
Tr\,}{\mathfrak A}({\mathbb M})$.

Next, let ${\rm w}={\rm w}_{1}+i{\rm w}_{2}$, where ${\rm
w}_{1},{\rm w}_{2}\in\mathfrak{A}_{\star}(\mathbb M)$ obey ${\rm
w}_{1,2}|_{\mathbb M_+}=w_{1,2}\circ\pi$ with $w_{1,2}\in\mathbb
A(M)$. Then one has
 $$
{\rm
w}|_{\Gamma_0^+}=[(w_{1}+iw_{2})|_{\Gamma_0}]\circ\pi\quad\text{and}\quad
{\rm
w}^{\star}|_{\Gamma_0^+}=[(w_{1}-iw_{2})|_{\Gamma_0}]\circ\pi\circ\tau.
 $$
Thus, the map $\eta\rightarrow\eta|_{\Gamma_{0}}$ is an isometric
isomorphism from ${\rm Tr\,}{\mathfrak A}({\mathbb M})$ onto ${\rm
Tr}_{\Gamma_{0}}\mathbb{A}(M)$. In particular, ${\rm
Tr}_{\Gamma_{0}}\mathbb{A}(M)$ is a Banach algebra, which is
isometrically isomorphic to $\mathbb{A}(M)$ via the map (\ref{tr
1}). Moreover, one has
$${\rm Tr\,}{\mathfrak A}({\mathbb M})={\rm clos}_{C(\partial\mathbb{M},\mathbb{C})}{\rm Tr\,}{\mathfrak A}^{\infty}({\mathbb M}),
\quad {\rm Tr}_{\Gamma_{0}}\mathbb{A}(M)={\rm
clos}_{C(\Gamma_{0},\mathbb{C})}{\rm
Tr}_{\Gamma_{0}}\mathbb{A}^{\infty}(M)$$ by the density of the
corresponding smooth subalgebras.
\smallskip

\noindent{$\bullet$}\,\,\,Summarizing and denoting by ${\mathfrak
A}\cong{\mathfrak B}$ the isometrically  isomorphic involutive
algebras, we arrive at the following scheme of determination of
${\mathfrak A}({\mathbb M})$ from the DN-map:
\begin{equation}\label{recover}
\Lambda_{g}^{\rm gr}\text{ or } \Lambda_{g}^{\rm
is}\overset{(\ref{tr alg via DNs})}\Rightarrow {\rm
Tr}_{\Gamma_{0}}\mathbb{A}^{\infty}(M)\overset{{\rm
clos}_{C(\Gamma_{0},\mathbb{C})}}{\Longrightarrow}{\rm
Tr}_{\Gamma_{0}}\mathbb{A}(M)\cong\mathbb{A}(M)\cong{\mathfrak
A}({\mathbb M})\,.
\end{equation}
Thus, $\Lambda_{g}$ determines the algebra
$\mathfrak{A}(\mathbb{M})$ up to an isometric isomorphism.

\subsubsection*{Determination of $\{M,g\}$}
\noindent$\bullet$\,\,\,Recall the well-known notions and facts
(see, e.g., \cite{Gamellin,Royden})

A {\it character} of the complex commutative Banach algebra
$\mathfrak{A}$ is a nonzero homomorphism
$\chi:\mathfrak{A}\to{\mathbb C}$. The set of all characters ({\it
spectrum} of $\mathfrak{A}$) is denoted by
$\widehat{\mathfrak{A}}$ and endowed with the canonical Gelfand
($*$-weak) topology.

The Gelfand transform ${\mathfrak A}\to
C(\widehat{\mathfrak{A}};{\mathbb C})$ maps $a\in{\mathfrak A}$ to
the function $\widehat{a}\in C(\widehat{\mathfrak{A}};\mathbb{C})$
by $\widehat{a}(\chi):=\chi(a)$. We write $S\cong T$ if the
topological spaces $S$ and $T$ are homeomorphic. If the algebras
are isometrically isomorphic (${\mathfrak A}\cong{\mathfrak B}$)
then their spectra are homeomorphic: $\widehat{\mathfrak
A}\cong\widehat{\mathfrak B}$.

For a function algebra $\mathfrak{F}\subset C(T;\mathbb{C})$, the
set $\mathscr{D}$ of the Dirac measures $\delta_t: a\mapsto
a(t),\,\,t\in T$ is a subset of $\widehat{\mathfrak{F}}$. This
algebra is called {\it generic} if
$\widehat{\mathfrak{F}}=\mathscr{D}$ holds, which is equivalent to
$\widehat{\mathfrak{F}}\cong T$. In this case, the algebra
$\mathfrak{F}$ is isometrically isomorphic to a subalgebra of
$C(\widehat{\mathfrak{F}};\mathbb{C})$, the isometry being
realized by the Gelfand transform $a\mapsto\widehat a$.

The algebra $\mathfrak{A}({\mathbb M})$ is generic
\cite{Royden,Stout}. Therefore, ${\mathfrak
A}(\mathbb{M})\cong{\rm Tr\,}{\mathfrak A}({\mathbb M})$ implies
 $$
\mathbb
{M}\cong\widehat{\mathfrak{A}(\mathbb{M})}\cong\widehat{{\rm
Tr}_{\Gamma_{0}}\mathbb{A}(M)}=:\mathbb{M}^{\,'}.
 $$
The fact of the crucial value for the subsequent is as follows. In
accordance with the scheme (\ref{recover}), each of the DN-maps
$\Lambda_{g}^{\rm gr}$ or $\Lambda_{g}^{\rm is}$ does determine
the spectrum ${\mathbb M}^{\,'}$.

The `space' involution $\tau$ on $\mathbb{M}$ induces the
involution $\tau'$ on the spectrum
$\widehat{\mathfrak{A}(\mathbb{M})}\cong \mathbb{M}^{\,'}$ by
$\tau':\delta_x\mapsto\delta_{\tau(x)}$, so that one has
\begin{equation}\label{inv ind}
[\tau'(\delta_{x})]({\rm w})=\delta_{\tau(x)}({\rm w})= {\rm
w}(\tau(x))=\overline{{\rm
w}^{\star}(x)}=\overline{\delta_{x}({\rm w}^{\star})}\quad
x\in\mathbb{M}, \ {\rm w}\in\mathfrak{A}(\mathbb{M}).
\end{equation}

\smallskip

\noindent$\bullet$\,\,\,Recall, what `to solve the EIT problem'
is. Assume that the DN-map $\Lambda_{g}^{\rm gr}$ (or
$\Lambda_{g}^{\rm is}$) of the unknown manifold $\{M,g\}$ is given
on $\Gamma_0\subset\partial M$. One needs to provide a manifold
$\{M',g'\}$ ({\it a copy} of the unknown $M$) such that $\partial
M'\supset\Gamma_0$ and $\Lambda_{g'}^{\rm gr}=\Lambda_{g}^{\rm
gr}$ (or $\Lambda_{g'}^{\rm is}=\Lambda_{g}^{\rm is}$) holds. We
claim and repeat that this is the only relevant understanding of
{\it to recover the unknown surface} \cite{BCald,B Sobolev Geom
Rings,B UMN}.
\smallskip

The required copy is constructed by means of the following
procedure. We describe it briefly, referring the reader to the
papers \cite{BCald,B Sobolev Geom Rings} for details. Note that,
starting the procedure, we have nothing but the operator
$\Lambda=\Lambda_{g}^{\rm gr}$ (or $\Lambda=\Lambda_{g}^{\rm is}$)
on $\Gamma_{0}$. However, we know a priori that $\Lambda$ is the
DN-map of some unknown Riemann surface.
\smallskip

{\it Step 1.}\,\,\,Given $\Lambda$, one recovers the metric $ds$
and the integration $J$ on $\Gamma_{0}$. Checking the relation
${\rm Ker\,}[I+(\Lambda J)^{2}]=\{0\}$, one detects the presence
(or absence) of other connected components of $\partial M$. We
also establish that $\Lambda=\Lambda_{g}^{\rm gr}$ if ${\rm Ker
\,}\Lambda=\{0\}$ or $\Lambda=\Lambda_{g}^{\rm is}$ if ${\rm
Ker\,}\Lambda=\{ {\rm const} \}$.
\smallskip

{\it Step 2.}\,\,\,The algebra $\{{\rm
Tr}_{\Gamma_{0}}\mathbb{A}(M),*\}$ is determined by (\ref{tr alg
via DNs}), the involution being determined by (\ref{invdef+}).
Then one finds its spectrum $\widehat{{\rm
Tr}_{\Gamma_{0}}\mathbb{A}(M)}=:\mathbb{M}^{\,'}$, which is a
homeomorphic copy of the unknown $\mathbb M$. Applying the Gelfand
transform
$${\rm Tr}_{\Gamma_{0}}\mathbb{A}(M)\ni \eta\,\mapsto\,\widehat{\eta}\in C(\mathbb{M}^{\,'};\mathbb{C})\,, $$
one gets the relevant copies $\widehat{\eta}$ of the (unknown)
holomorphic functions ${\rm w}^{\eta}$, which are supported on
${\mathbb M}$ and have the given traces ${\rm
w}^{\eta}|_{\Gamma_{0}}=\eta$.
\smallskip

{\it Step 3.}\,\,\,Using $\Re\widehat{\eta}$ and
$\Im\widehat{\eta}$ in capacity of the local coordinates on
${\mathbb M}^{\,'}$, one endows $\mathbb{M}^{\,'}$ with the
structure of a smooth 2d manifold.
\smallskip

{\it Step 4.}\,\,\, By Shilov \cite{Gamellin,BCald}, the boundary
$\partial {\mathbb M}^{\,'}$ is determined as the subset of
${\mathbb M}^{\,'}$, at which the functions $|\widehat\eta|$
attain the maximum. This boundary consists of two disjoint
connected components, which corresponds to the fact that
$\partial\mathbb M$ is $\Gamma^+_0\cup\Gamma^-_0$.

Then we identify the boundary points by
$$\partial{\mathbb M}\ni x\equiv\delta\in\partial {\mathbb M}^{\,'}\,\Leftrightarrow\,
\,{w}(x)=\widehat{w}(\delta)\quad\text{for all smooth}\,\,w $$
and, thus, attach $\partial {\mathbb M}^{\,'}$ to
$\partial{\mathbb M}$. Since the metric $ds$ on $\Gamma_{0}$ is
known, the mutual boundary $\partial {\mathbb M}\equiv\partial
{\mathbb M}^{\,'}$ is also endowed with the metric $ds'\equiv ds$.
\smallskip

{\it Step 5.}\,\,\, Endow ${\mathbb M}^{\,'}$ with the involution
$\tau'$ by the rule
 $$
[\tau'(\chi)](\eta):=\overline{\chi(\eta^*)},\qquad\eta\in {\rm
Tr}_{\Gamma_0}\mathbb A(M),\,\,\chi\in \mathbb M'.
 $$
It is a copy of the involution $\tau$ in ${\mathbb M}$ in view of
(\ref{inv ind}). By the same (\ref{inv ind}), the copy of
$\tilde{\Gamma}$ is the subset $\tilde{\Gamma'}$ of hermitian
characters in ${\mathbb M}^{\,'}$.
\smallskip

{\it Step 6.}\,\,\,At the moment, there is no metric on the
spectrum $\mathbb{M}^{\,'}$. In the meantime, it supports the
reserve of functions $\Re\widehat{\eta}$, $\Im\widehat{\eta}$,
which are the relevant copies of the (unknown) harmonic functions
$\Re{\rm w}_{\eta}$, $\Im{\rm w}_{\eta}$ on ${\mathbb M}$. As is
known, this reserve determines a metric $\tilde{\rm g}$ on
$\mathbb{M}^{\,'}$, which provides $\Delta_{\tilde{\rm
g}}\Re\widehat{\eta}=\Delta_{\tilde{\rm g}}\Im\widehat{\eta}=0$,
such a metric being determined up to a conformal deformation. In
\cite{LUEns,BCald} the reader can find concrete tricks for
determination of $\tilde{\rm g}$ (see also \cite{BV_CUBU_2}, page
16). One of them is to write the equations $\Delta_{\tilde{\rm
g}}\Re\widehat{\eta}=0$, $\Delta_{\tilde{\rm g}}\Im\hat{\eta}=0$
for a rich enough set $\eta=\eta_{1},\dots,\eta_n$ in local
coordinates and then use these equations as a system for finding
$\tilde{\rm g}_{ij}$ (up to a Lipshitz functional factor). Also,
it is easily seen that one can choose the metric $\tilde{\rm g}$
to be $\tau'$--invariant, i.e., obeying $\tau'_*\tilde{\rm
g}=\tilde{\rm g}$.

Let such a metric $\tilde{\rm g}$ be chosen. Find a smooth
positive function $\rho$ on ${\mathbb M}^{\,'}$ provided
$\rho=\rho\circ\tau'$ and such that the length element $ds''$ of
the metric $g''=\rho\tilde{\rm g}$ at $\partial {\mathbb M}^{\,'}$
coincides with the (known) element $ds'$.
\smallskip

{\it Step 7.}\,\,\, Denote by $M''$ some connected component of
$\mathbb{M}^{\,'}\backslash\tilde{\Gamma'}$; then
$M':=M''\cup\tilde{\Gamma'}$ is a homeomorphic copy of the unknown
original $M$. Denoting by $g'$ the restriction of the metric
$\tilde{\rm g}$ on $M'$, we obtain the manifold $(M',g')$, which
satisfies $\partial M'=\Gamma$ and $\Lambda_{g'}^{\rm
gr}=\Lambda_{g}^{\rm gr}$, $\Lambda_{g'}^{\rm is}=\Lambda_{g}^{\rm
is}$ by construction. The latter solves the EIT problem.
\bigskip

It is worth noting the following. In the course of solving the
problem by this procedure, we operate not with the attributes of
$M$ themselves (which is impossible in principle!), but create
their copies, and recover not the original $M$, but its relevant
copy $M'$. This view of what is happening is fully consistent with
the philosophy of the BC-method in inverse problems \cite{BCald,B
UMN}: the only reasonable understanding of `to restore unreachable
object' is to construct its (preferably, isomorphic) copy. Such a
point of view is supported by very general principles of the
system theory: see \cite{KFA}, Chapter 10.6, Abstract Realization
theory.


\begin{thebibliography}{99}

\bibitem{Aless1}
Giovanni Alessandrini, Luca Rondi
\newblock {Optimal Stability for the Inverse Problemof Multiple Cavities.}
\newblock {\em Journal of Differential Equations}, Volume 176, Issue 2, (2001) 356--386.
https://doi.org/10.1006/jdeq.2000.3987.

\bibitem{Aless2}
Giovanni Alessandrini, Elena Beretta, Edi Rosset, Sergio Vessella
\newblock {Optimal stability for inverse elliptic boundary value problems with
unknown boundaries.}
\newblock {\em Annali della Scuola Normale Superiore di Pisa - Classe di
Scienze}, Serie 4, Tome 29 (2000) no. 4, pp. 755--806.
http://www.numdam.org/item/ASNSP-2000-4-29-4-755-0.

\bibitem{BCald}
M.I.Belishev.
\newblock {The Calderon problem for two-dimensional manifolds
by the BC-method.}
\newblock {\em SIAM Journal of Mathematical Analysis}, 35, no 1:
172--182, 2003.

\bibitem{B Sobolev Geom Rings}
M.I.Belishev. Geometrization of Rings as a Method for Solving
Inverse Problems. {\em Sobolev Spaces in Mathematics III.
Applications in Mathematical Physics, Ed. V.Isakov.}, Springer,
2008, 5--24.

\bibitem{B UMN}
M.I.Belishev. {\newblock Boundary Control and Tomography of
Riemannian Manifolds.}
\newblock {\em Russian Mathematical Surveys}, 2017, 72:4, 581--644.
doi.org/10.4213/rm 9768.

\bibitem{BKor_JIIPP}
M.I.Belishev, D.V.Korikov.
\newblock {On the EIT problem for nonorientable
surfaces.}
\newblock {\em
Journal of Inverse and Ill-posed Problems Journal of Inverse and
Ill-posed Problems}; 18 December 2020. doi:10.1515/jiip-2020-0129

\bibitem{BKor_char_arXiv}
M.I.Belishev, D.V.Korikov.
\newblock {On characterization of Dirichlet-to-Neumann map of Riemannian
surface with boundary.}
\newblock {\em arXiv:2103.03944v1 [math.AP] 5 Mar 2021}.

\bibitem{BV_CUBU_2}
M.I.Belishev and A.F.Vakulenko.
\newblock {On algebraic and uniqueness properties of 3d harmonic quaternion
fields}.
\newblock {\em CUBO A Mathematical Journal}, 21, No 01, April 2019,
pp 01--19. http://dx.doi.org/10.4067/S0719-06462019000100001.

\bibitem{Gamellin}
T.W.Gamellin.
\newblock {Uniform Algebras.}
\newblock {\em AMS Chelsea Publishing}, AMS, Providence, Rhod Island, 2005.

\bibitem{HMich}
G.Henkin, V.Michel.
\newblock {On the explicit reconstruction of a Riemann surface from
its Dirichlet-Neumann operator.}
\newblock {\em Geometry and Fanctional Analysis}, 17 (2007), no 1, 116--155.

\bibitem{KFA}
R.Kalman, P.Falb, M.Arbib. Topics in Mathematical System Theory.
\newblock{\em New-York: McGraw-Hill}, 1969.

\bibitem{LUEns}
M.Lassas, G.Uhlmann.
\newblock {On determining a Riemannian manifold from the
Dirichlet-to-Neumann map.}
\newblock {\em Ann. Scient. Ec. Norm. Sup.}, 34(5): 771-- 787, 2001.

\bibitem{Royden}
H.I.Royden.
\newblock {Function algebras.}
\newblock {\em Bulletin of the American Mathematical Society}, 69, No 3 (1963),
281--298.

\bibitem{Stout}
E.L.Stout.
\newblock {Two theorems concerning functions holomorphic on multy connected domains.}
\newblock {\em Bulletin of the American Mathematical Society}, 69, (1963),
527--530. MRD150274

\bibitem{Taylor}
M.E.Taylor.
\newblock {Partial Differential Equations II}.
\newblock {\em Qualitative Studies
of Linear Equations}, Second Edition, Applied Mathematical
Sciences, Volume 116.

\bibitem{Vekua}
I. Vekua.
\newblock {Generalized Analytic Functions}.
\newblock {\em Pergamon Press}, 1962.
\end{thebibliography}
\end{document}